\documentclass[10pt,letterpaper,conference]{IEEEtran}

\usepackage[T1]{fontenc}
\usepackage{amsmath,amssymb,amsfonts,amsthm}
\usepackage{graphicx,epsfig,psfrag,color}
\usepackage{times}
\usepackage{dsfont}
\usepackage{algorithm}
\usepackage[noend]{algorithmic}

\newcommand{\bi}{\begin{itemize}}
\newcommand{\ei}{\end{itemize}}
\newcommand{\bal}{\begin{align}}
\newcommand{\eal}{\end{align}}
\newcommand{\EE}{\mathbb{E}}

\DeclareMathOperator*{\argmax}{arg\,max}
\DeclareMathOperator*{\argmin}{arg\,min}

\newtheorem{theorem}{Theorem}
\newtheorem{lemma}{Lemma}

\newtheorem{note}{Note}

\newtheorem{proposition}{Proposition}

\begin{document}

\vspace{-.1in}
\title{\huge{From the Information Bottleneck to the Privacy Funnel}}

\vspace{-.1in}
 \author{\IEEEauthorblockN{Ali Makhdoumi}
\IEEEauthorblockA{
MIT\\
Cambridge, MA\\
makhdoum@mit.edu}\\
\and
\IEEEauthorblockN{Salman Salamatian}
\IEEEauthorblockA{
EPFL\\
Lausanne, Switzerland\\
salman.salamatian@epfl.ch}\\
\and
\IEEEauthorblockN{Nadia Fawaz}
\IEEEauthorblockA{Technicolor\\
Palo Alto, CA\\
nadia.fawaz@technicolor.com}\\
\and
\IEEEauthorblockN{Muriel M{\'e}dard}
\IEEEauthorblockA{\
MIT\\
Cambridge, MA\\
medard@mit.edu}
}

\maketitle
\begin{abstract}
We focus on the privacy-utility trade-off encountered by users who wish to disclose some information to an analyst, that is correlated with their private data, in the hope of receiving some utility. 
We rely on a general privacy statistical inference framework, under which data is transformed before it is disclosed, according to a probabilistic privacy mapping.
We show that when the log-loss is introduced in this framework in both the privacy metric and the distortion metric, the privacy leakage and the utility constraint can be reduced to the mutual information between private data and disclosed data, and between non-private data and disclosed data respectively. 
We justify the relevance and generality of the privacy metric under the log-loss by proving that the inference threat under any bounded cost function can be upperbounded by an explicit function of the mutual information between private data and disclosed data.
We then show that the privacy-utility tradeoff under the log-loss can be cast as the non-convex \emph{Privacy Funnel} optimization, and we leverage its connection to the Information Bottleneck, to provide a greedy algorithm that is locally optimal. We evaluate its performance on the US census dataset.
%
\end{abstract}

\section{Introduction}

We consider a setting in which users have two kinds of data, that are correlated: some data that each user would like to remain private and some non-private data that he is willing to disclose to an analyst and from which he will derive some utility.
The analyst is a legitimate receiver of the disclosed data, which he
will use to provide utility to the user, but can also adversarially
exploit it to infer the user's private data. 
This creates a tension between privacy and utility requirements. 
To reduce the inference threat on private data while maintaining utility, each user's non-private
data is transformed before it is disclosed, according to a probabilistic privacy mapping.
The design of the privacy mapping should balance the tradeoff between the utility of the disclosed data, and the privacy of the private data: it should keep the disclosed transformed data as much informative as possible about the non-private data, while leaking as little information as possible about the private data. 

The framework for privacy against inference attacks in \cite{du2012privacy} proposes to design the privacy mapping as the solution to an optimization minimizing the inference threat subject to a utility constraint. Our approach relies on this framework, and makes the following two  contributions.
%
First, we show that when the log-loss is introduced in this framework in both the privacy metric and the distortion metric, the privacy leakage reduces to the mutual information between private data and disclosed data, while the utility requirement is modeled by the mutual information between non-private data and disclosed data. We justify the relevance and generality of the privacy metric under the log-loss by proving that the inference threat, defined in \cite{du2012privacy} as the inference cost gain, under any bounded cost function can be upperbounded by an explicit function of the mutual information between private data and disclosed data. We then show that the privacy-utility tradeoff under the log-loss can be cast as the \emph{Privacy Funnel} optimization, and study its connection to the Information Bottleneck \cite{tishby2000information}.
%
Second, for general distributions, the privacy funnel optimization being a non-convex problem, we provide a greedy algorithm  for the Privacy Funnel that is locally optimal by leveraging connections to the Information Bottleneck method \cite{aggloinfobottleneck,tishby2000information}, and evaluate its performance on real-world data.
%
	
%

\noindent\textbf{Related Work}: Several works, such as \cite{yamamoto1983source,agrawal2000privacy,dwork2006calibrating,dwork2006differential,Sankar-Poor-IFStrans2013}, have studied the issue of keeping some information private while disclosing some correlated information, by distorting the information disclosed. Differential privacy \cite{dwork2006calibrating,dwork2006differential} was introduced to answer queries on statistical databases in a privacy-preserving manner, by minimizing the chances of identification of the database records.
One line of work in information theoretic privacy \cite{yamamoto1983source,Sankar-Poor-IFStrans2013} studies the trade-off between
privacy and utility, where they consider expected distortion as a measure of utility and equivocation as a measure of privacy. \cite{Sankar-Poor-IFStrans2013} focus mainly on collective privacy for all or subsets of the entries of a database, and provide fundamental and asymptotic results on the rate-distortion-equivocation region as the number of data samples grows arbitrarily large. These approaches are different from our approach in three ways as we do not consider a communication problem where the rate needs to be bounded, 
and we use the average amount of bits as a measure of both utility and privacy (log-loss distortion or mutual information).  
The wire-tap channel, introduced in \cite{wyner1975wire}, focuses on designing
the encoder and decoder to release information and protect private information from an eavesdropper, when utility is measured in terms of error probability of the decoded message, and secrecy is measured in terms of normalized mutual information.  
Our setting differs from the wire-tap channel, as it does not involve a third-party eavesdropper, but the analyst is both a legitimate receiver of the disclosed data, and a potential adversary as it can use it to try to infer private data. Moreover, we focus on the privacy mapping design (channel design), with different measures of privacy and utility. 

The log-loss distortion  has been studied in \cite{courtade2011multiterminal} as a measure of distortion in the context of multi-terminal source coding. Log-loss as measure of distortion is also studied in  \cite{harremoes2007information} where they show that log-loss satisfies certain properties that leads to the Information Bottleneck method \cite{tishby2000information}. Finally, for an overview of the central role of the log-loss distortion in prediction, we refer the reader to \cite{merhav1998universal}.

\noindent\textbf{Outline}: In Section~\ref{sec:problemformandcharac}, we
introduce the privacy-utility trade-off against Inference attacks. In Section \ref{sec:Funnel}, we describe the privacy funnel method and show   properties of log-loss metric, and then characterize the privacy-disclosure trade-off as the privacy  funnel optimization.  In Section~\ref{sec:alg}, we provide a greedy algorithm to design the privacy mapping and
evaluate it on real-world data.

\noindent \textbf{Notations}: Throughout the paper, $X$ denotes a random variable over alphabet $\mathcal{X}$ with distribution $P_X$.
All random variables are assumed to be discrete, unless mentioned otherwise.


\section{Privacy-Utility against Inference attacks}
\label{sec:problemformandcharac}
In this background section, we first describe the setting, and the privacy and utility metrics introduced in the framework for privacy against inference attacks in \cite{du2012privacy}. Then, we recall how the privacy-utility trade-off can be cast into an optimization.

\vspace{-.05in}
\subsection{Setting}
\vspace{-.05in}


We consider a setting 
where a user has some private data, represented by the random variable $S \in \mathcal{S}$, which is correlated with some non-private data $X\in \mathcal{X}$, that the user wishes to share with an analyst. The correlation between $S$ and $X$ is captured by the joint distribution $P_{S, X}$. Due to this correlation, releasing $X$ to the analyst would enable him to draw some inference on the private data $S$.
To reduce the inference threat on $S$ that would arise from the observation of $X$, rather than releasing $X$, the user releases a distorted version of $X$ denoted by $Y\in \mathcal{Y}$. The distorted data $Y$ is generated by passing $X$ through a conditional distribution $P_{Y|X}$, called the privacy mapping. Throughout the paper, we assume $S \to X \to Y$ form a Markov chain.
Therefore, once the distribution $P_{Y|X}$ is fixed, the joint distribution $P_{S, X, Y}= P_{Y|X} P_{S, X}$ is defined.

The analyst is a legitimate recipient of data $Y$, which it can use to provide utility to the user, e.g. some personalized service. However, the analyst can also act as an adversary by using $Y$ to illegitimately infer private data $S$. The privacy mapping aims at balancing the tradeoff between utility and privacy: the privacy mapping should be designed to decrease the inference threat on private $S$ by reducing the dependency between $Y$ and $S$, while at the same time preserving the utility of $Y$, by maintaining the dependency between $Y$ and $X$.

%
%
%

\vspace{-.05in}
\subsection{Privacy Metric}

We consider the inference threat model introduced in \cite{du2012privacy}, in which the analyst performs an adversarial inference attack on the private data $S$. More precisely, the analyst selects a distribution $q$, from the set $\mathcal{P}_S$ of all probability distributions over $\mathcal{S}$, that minimizes an expected inference cost function $C(S, q)$. In other words, the analyst chooses in an adversarial way a belief distribution $q$ over the
private variables $S$ prior to observing $Y$, and a revised belief distribution as 
\vspace{-.02in}
\begin{equation}
q_0^* = \underset{q \in \mathcal{P}_S}{\argmin \;} \mathbb{E}_{P_S} [C(S,q)], \nonumber
\end{equation}
\vspace{-.02in}
prior to observing $Y$, and a revised belief distribution
\vspace{-.02in}
\begin{equation}
q_y^* = \underset{q \in \mathcal{P}_S}{\argmin \;} \mathbb{E}_{P_{S|Y}} [C(S,q)|Y=y], \nonumber
\end{equation}
\vspace{-.02in}
after observing $Y=y$.
This models a very broad class of adversaries that perform statistical inference.
Using the chosen belief distribution $q$, the analyst can produce an estimate
of the input $S$, e.g. using a Maximum a Posteriori (MAP) estimator.
Let $c_0^*$ and $c_y^*$ respectively denote the minimum average cost of inferring $S$ without observing $Y$, and after observing $Y=y$:
\begin{align*}
& c_0^* = \underset{q \in \mathcal{P}_S}{\min \;} \mathbb{E}_{P_S} [C(S,q)],\, \,  &
 c_y^* = \underset{q \in \mathcal{P}_S}{\min \;} \mathbb{E}_{P_{S|Y}} [C(S,q)|Y=y].
\end{align*}

Thanks to the observation of $Y$, the analyst obtains an average gain in inference cost of
$\Delta C = c_0^* - \mathbb{E}_{P_{Y}} [c_Y^*]$.
The average inference cost gain $\Delta C$ was proposed as a general privacy metric in \cite{du2012privacy}, as it measures the improvement in the quality of the inference of private data $S$ due to the observation of~$Y$.
The design of the privacy  mapping $P_{Y|X}$ should aim at reducing $\Delta C$, or in other words it should aim at bringing the inference cost given the observation of $Y$ closer to the initial inference cost $c_0^*$ without observing $Y$. 

\vspace{-.02in}
\subsection{Accuracy Metric}

The privacy mapping should maintain the utility of the distorted data $Y$. In the framework proposed in \cite{du2012privacy}, the utility requirement is modeled by a constraint on the average distortion $\EE_{P_{X,Y}}[d(X, Y)] \le D$, for some distortion measure $d: \mathcal{X} \times \mathcal{Y} \rightarrow \mathds{R}^+ $, and some distortion level $D \geq 0$. Assuming that the distortion measure $d$ is a function of $X$ and $Y$, but not of their statistical properties, the average distortion $\EE_{P_{X,Y}}[d(X, Y)]$ is linear in $P_{Y|X}$. Consequently, the distortion constraint is a linear constraint in $P_{Y|X}$.

\subsection{Privacy-Accuracy tradeoff}\label{sec:originalTradeoff}

The optimal privacy mapping for a given distortion level $D$ is obtained as the solution of the following optimization 
 \vspace{-.03in}
\begin{align}\label{Eq:optmain}
& \min_{P_{Y|X}~:~ \EE_{P_{X,Y}}[d(X, Y)] \le D} \Delta \, C
\end{align}
If $\Delta \, C$ is convex in $P_{Y|X}$, then optimization~\eqref{Eq:optmain} is a convex optimization, since the distortion constraint  $\EE_{P_{X,Y}}[d(X, Y)]$ is linear in $P_{Y|X}$. 

\section{The Privacy Funnel Method}\label{sec:Funnel}
In this section, we focus on the privacy-utility framework when the log-loss is used in both the privacy metric and in the distortion metric. We justify the relevance of the log-loss in such a framework, and characterize the resulting privacy-disclosure tradeoff as the Privacy Funnel optimization. Finally, we show how the Privacy Funnel is related to the Information Bottleneck \cite{tishby2000information}, and how algorithms developed for the latter can inform the design of algorithms for the former.

\subsection{Privacy metric under log-loss}

In this section, we focus on the threat model under the log-loss cost function. We first recall that, under this cost-function, the privacy leakage can be measured by the mutual information $I(S; Y)$ between the private variable $S$ and the variable $Y$. We then justify the relevance and the generality of the use of the log-loss in the threat model, by showing that the inference cost gain for any bounded cost function can be upperbounded by a function of the mutual information between $S$ and $Y$.

Under the log-loss cost function $C(s,q)= - \log q(s)$, $\forall s\in \mathcal{S}$, the privacy leakage can be measured by the mutual information $I(S; Y)$, as stated in the following lemma.
%
\begin{lemma}[\cite{du2012privacy}]\label{Lem:MutualInfoThreat}
The average inference cost gain under the log-loss cost function $C(s,q)= - \log q(s)$, is
the mutual information between $S$ and $Y$: $\Delta C= I(S; Y)$.
\end{lemma}
%
\begin{proof}
Let the cost function to be the log-loss defined by $C(s,q)= - \log q(s)$ for any $s\in \mathcal{S}$. Then,
\begin{equation}
 c_0^* = \underset{q \in \mathcal{P}_S}{\min \;} \mathbb{E}_{P_S} [-\log q(S)]= \mathbb{E}_{P_S} [-\log p(S)]+ D(p||q).
\end{equation}
Since $D(p||q) \ge 0$, with equality if $p = q$, we have $c_0^*= H(S)$. Similarly,  we have $c_y^*=H(S|Y=y)$. Therefore, $\Delta C = H(S) - \mathbb{E}_{P_{Y}} [H(S|Y=y)] = I(S; Y)$.
\end{proof}
We now justify the relevance and the generality of the use of the log-loss in the threat model. More precisely, in Theorem~\ref{Th:MIboundscost} below, we prove that for any bounded cost function $C(S,q)$, the associated inference cost gain $\Delta C$  can be upperbounded by an explicit constant factor of $\sqrt{I(S; Y)}$. Thus, controlling the cost gain under the log-loss, so that it does not exceed a target privacy level, is sufficient to ensure that the privacy threat under a different bounded cost function would also be controlled. Therefore, the design of the privacy mapping can be focused on minimizing the privacy leakage as measured by $I(S; Y)$.
\vspace{-.02in}
\begin{theorem}\label{Th:MIboundscost}
Let $L=\sup_{s\in \mathcal{S}, q\in \mathcal{P}_S} |C(s,q)|< \infty$. We have
 $\Delta C = c_0^* - \mathbb{E}_{P_Y} [c_Y^*] \leq 2 \sqrt{2} L \sqrt{ I(S;Y)}$.
\end{theorem}

The proof of Theorem~\ref{Th:MIboundscost} requires  the following lemma.
\vspace{-.02in}
\begin{lemma}\label{Lem:costlogloss}
Let $C(s, q)$ be a bounded cost function such that $L=\sup_{s\in \mathcal{S}, q\in \mathcal{P}_S} |C(s,q)|< \infty$. For any given $y \in \mathcal{Y}$,  
\begin{align*}
\mathbb{E}_{P_{S|Y}} [ C(S,q_0^*) - C(S,q_y^*)|Y=y] \le 2 \sqrt{2} L \sqrt{ D(P_{S|Y=y}|| P_S)}. \nonumber
\end{align*}

\end{lemma}
%
\begin{proof}
we have
\begin{align}
& \mathbb{E}_{P_{S|Y}} [ C(S,q_0^*) - C(S,q_y^*)|Y=y]\nonumber \\
& = \sum_{s} p(s|y) [ C(s, q_0^* ) - C(s, q_y^*) ] \nonumber\\
& = \sum_{s} (  p(s|y) - p(s)  + p(s) ) [ C(s, c_0^* ) - C(s, q_y^*) ] \nonumber \\
& = \sum_{s} ( p(s|y) - p(s) ) [C(s, q_0^* ) - C(s, q_y^*)  ] \nonumber\\
& + \sum_{s}p(s)[C(s, q_0^* ) - C(s, q_y^*) ]  \nonumber \\
& \le 2L \sum_{s}  |p(s|y) - p(s)|+ (\mathbb{E}_{P_S}[C(S, q_0^*)]- \mathbb{E}_{P_S}[C(S, q_y^*)] ), \nonumber\\
& \le 2L \sum_{s}  |p(s|y) - p(s)|, \nonumber\\
& = 4 L || P_{S|Y=y} - P_S||_{TV} \nonumber\\
& \le 4 L \sqrt{\frac{1}{2} D(P_{S|Y=y}|| P_S)} \nonumber ,
\end{align}
where we used that $C(s, q_0^* ) - C(s, q_y^*) \le 2L$ and  $ \mathbb{E}_{P_S}[C(S, q_0^*)]- \mathbb{E}_{P_S}[C(S, q_y^*)]  \le 0 $. And the last inequality follows from using Pinsker's inequality (where the log in the definition of divergence is natural log).
\end{proof}

We now prove Theorem~\ref{Th:MIboundscost}.
\begin{proof}[proof of Theorem~\ref{Th:MIboundscost}]
\noindent We have
\begin{align}
&\Delta C = \mathbb{E}_{P_S} [C(S,q_0^*)] - \mathbb{E}_{P_Y} \left[ \mathbb{E}_{P_{S|Y}} [C(S,q_y^*)|Y=y]\right] \nonumber \\
& = \mathbb{E}_{P_Y} \left[ \mathbb{E}_{P_{S|Y}} [C(S,q_0^*) - C(S,q_y^*)|Y=y ] \right]  \nonumber \\
& \le 2 \sqrt{2} L \mathbb{E}_{P_Y} \left[ D(P_{S|Y=y}|| P_S) \right] \le  2 \sqrt{2} L \sqrt{ I(S; Y) }, \nonumber
\end{align}
\vspace{-.02in}
where the last step follows from concavity of
square root function and the one before that follows from Lemma \ref{Lem:costlogloss}.
\end{proof}

\vspace{-.02in}
\subsection{Accuracy metric under log-loss}
\vspace{-.02in}

Consider the log-loss distortion defined as  $d(x, y)= - \log {P(X=x|Y=y)}$, which is a function of $x$ and $y$ as well as $P_{Y|X}$. Using log-loss, the average distortion is $\EE[d(X, Y)] = \EE_{P_{X, Y}}[-\log P_{X|Y}]   = H(X|Y)$ that can be minimized by designing the mapping $P_{Y|X}$.
Thus, the constraint $\mathbb{E}[d(X, Y)] \le D$ would be $H(X|Y) \le D$ for a given distortion level, $D$. Given $P_X$, and therefore $H(X)$, and assuming that $R= H(X)-D$, the distortion constraint can be rewritten as $I(X; Y) \ge R$, that is the same as the constraint of \eqref{eq:optdkl}. It should be noted that the average distortion under the log-loss is not linear in $P_{Y|X}$.

For a given $P_{SX}$ and $P_{Y|X}$, where $S\to X\to Y$, we define the \emph{disclosure} to be the mutual information between $X$ and~$Y$.

\vspace{-.02in}
\subsection{Privacy-Disclosure Trade-off}\label{sec:tradeoff}



There is a trade-off between the information that user shares about $X$ and the information that user keeps private about $S$. We pass $X$ through a randomized mapping $P_{Y|X}$ and reveal $Y$ to the analyst. The purpose of this mapping is to make $Y$ informative about $X$ and to make $Y$ uninformative about $S$. Given $P_{SX}$, we design the privacy-mapping $P_{Y|X}$ to maximize the amount of information $I(X; Y)$ that user disclose about the public information, $X$, while  minimizing the  collateral information about the private variable $S$ measured by $I(S; Y)$.


 The trade-off between disclosure and privacy in the design of the privacy mapping is represented by the following optimization, that we refer to as the {\em Privacy Funnel}:
\vspace{-.05in}
\begin{align} \label{eq:optdkl}
&\min_{P_{Y|X}: ~ I(X; Y) \ge  R} I(S; Y).
\end{align}
For a given disclosure level $R$, among all feasible privacy mappings $P_{Y|X}$ satisfying $I(X; Y)\ge R$, the privacy funnel  selects the one that minimizes $I(S; Y)$.
%
Note that $I(X; Y)$ is convex in $P_{Y|X}$ and  since $P_{Y|S}$ is linear in $P_{Y|X}$ and
$I(S;Y)$ is convex in $P_{Y|S}$, the objective function $I(S;Y)$ is convex in $P_{Y|X}$. 
However, because of the constraint $I(X; Y)\ge R$, the Privacy Funnel~\eqref{eq:optdkl} is not a convex optimization \cite[Chap.~4]{boyd2004convex}.

\subsection{Connection to the Information Bottleneck Method}

The information bottleneck method, introduced in \cite{tishby2000information}, considers the setting where a variable $X$ is to be compressed, while maintaining the information it bears about another correlated variable $S$. The information bottleneck method is a technique generalizing rate-distortion, as it seeks to optimize the tradeoff between the compression length of $X$ and the accuracy of the information preserved about $S$ in the compressed output $Y$. The information bottleneck optimization  \cite{tishby2000information} is
\vspace{-.03in}
\begin{equation}\label{eq:bottleneckOptim}
\min_{{P_{Y|X}:~ I(S; Y) \ge C}} I(X; Y)
\end{equation}
\vspace{-.03in}
for some constant $C$.
In the information bottleneck, the compression mapping $P_{Y|X}$ is designed to make $X$ and $Y$ as far as possible from each other (minimizes $I(X; Y)$) while guaranteeing that $S$ and $Y$ are close to each other. In other words, in the information botteleneck the mapping $P_{Y|S}$ is designed to make $I(S; Y)$ large and $I(X; Y)$ small. The information bottleneck optimization~(\ref{eq:bottleneckOptim}) bears some resemblance to the privacy funnel~(\ref{eq:optdkl}), but is actually the opposite optimization. Indeed, in the privacy funnel, the privacy mapping is designed to make $I(S; Y)$ small and $I(X; Y)$ large. 

Several techniques were developed to solve the information bottleneck problem such as alternating iteration  \cite{tishby2000information} and agglomerative information bottleneck \cite{aggloinfobottleneck}. A question we examined is whether algorithms developed to solve the information bottleneck optimization could be adapted to solve the privacy funnel optimization.
The alternating iteration algorithm \cite{tishby2000information} finds a  stationary point of the Lagrangian of information bottleneck optimization~(\ref{eq:bottleneckOptim}) defined as $\mathcal{L}=I(X; Y)- \beta I(S; Y)$ for some $\beta$. The stationary point can be a local minimum,  which addresses the information bottleneck, or a local maximum  in which case it addresses the privacy funnel. However, there is no guarantee on the convergence of this alternating algorithm to either a local minimum or a local maximum. Thus, we developed a new greedy algorithm that is guaranteed to converge to a solution of the privacy funnel, which is the object of Section~\ref{sec:alg}. 





\section{Algorithm for the Privacy Funnel}\label{sec:alg}

\vspace{-.05in}
We showed that the privacy funnel~\eqref{eq:optdkl} optimization is not a convex optimization. In this section, we provide a greedy algorithm to solve this optimization and we evaluate it
on real-world data.
\subsection{Greedy Algorithm} Suppose the constraint $I(X; Y) \ge R$ is given for some $R \le H(X)$. We wish to find $P_{Y|X}$ that minimizes $I(S; Y)$. Note that for $\mathcal{Y}=\mathcal{X}$ and $P_{Y|X}(y|x)=\mathbf{1}\{x=y\}$ (where $\mathbf{1}\{x=y\}=1$ if and only if $x=y$), the condition
$I(X; Y) \ge R$ is satisfied because $I(X; Y)=H(X) \ge R$. However, $I(S; Y)$ might be too large. The idea is to merge two elements of
$\mathcal{Y}$ to make $I(S; Y)$ smaller, while satisfying $I(X; Y) \ge R$. This method is motivated by agglomerative
 information method introduced in \cite{aggloinfobottleneck}. We merge $y_i$ and $ y_j$ and denote the merged
 element by ${y_{ij}}$. We then update $P_{Y|X}$ as
$ p(y_{ij}|x)= p(y_i|x) + p(y_j|x), ~ \text{ for all } x\in \mathcal{X}$.
After merging, we also have
$p(y_{ij})= p(y_i)+ p(y_j)$.
Consider the row stochastic matrix $P$ as $P_{x, y}=P_{Y|X}(y|x)$ for all $x\in \mathcal{X}$ and all $y\in \mathcal{Y}$. In Algorithm (1) we start with $P$  as an identity matrix and then at each iteration we delete two columns of $P$  (corresponding to $y_i$ and $y_j$) and add their summation as a new column (corresponding to $y_{ij}$) to $P$. Thus, the resulting matrix  at the end contains only zeros and ones, determining all $x \in \mathcal{X}$ and all $y \in \mathcal{Y}$ such that $P_{Y|X}(y|x)=1$. Let $Y^{i-j}$ be the resulting $Y$ from merging $y_i$ and $y_j$. Algorithm $(1)$ is a greedy algorithm that uses this idea in order to solve optimization  \eqref{eq:optdkl}. One need to calculate $I(S; Y) - I(S; Y^{i-j})$ and $I(X; Y) - I(X; Y^{i-j})$ at each iteration of Algorithm (1). Proposition \ref{Pro:JSD} shows an efficient way to calculate them.

\begin{algorithm}[t] \label{alg:greedy}
\caption{Greedy algorithm-privacy funnel}
\begin{algorithmic}
\STATE \textbf{Input:} $R$, $P_{S, X}$
\STATE \textbf{Initialization: } $\mathcal{Y} = \mathcal{X}$, $P_{Y|X}(y|x)=\mathbf{1}\{y=x\}$.
\WHILE{ there exists $i', j'$ such that $I(X; Y^{i'-j'}) \ge R$}

        \STATE among those $i', j'$, let \\
         $\{y_i, y_j\}= \argmax_{y_{i'}, y_{j'} \in \mathcal{Y}}  I(S; Y) - I(S; Y^{i'-j'}) $
        \STATE    \textbf{merge}: $\{y_i, y_j\} \rightarrow y_{ij}$
    \STATE    \textbf{update}: $\mathcal{Y}= \{\mathcal{Y}\setminus \{y_i, y_j\} \}\cup \{y_{ij}\} $  and $P_{Y|X}$
\ENDWHILE
\STATE \textbf{Output:} $P_{Y|X}$
\end{algorithmic}
\vspace{-.06in}
\end{algorithm}
\vspace{-.05in}

\vspace{-.02in}
\begin{proposition}\label{Pro:JSD}
For a given joint distribution $P_{S, X, Y}= P_{S, X} P_{Y|X}$, we have $I(S; Y)-I(S; Y^{i-j})=$
\begin{align}
& p(y_{ij})
 H\left(\frac{p(y_i)P_{S|Y=y_j}+p(y_j)P_{S|Y=y_j}}{p(y_{ij})} \right) \nonumber\\
& - \left( p(y_i) H(P_{S|Y=y_i})+ p(y_j) H(P_{S|Y=y_j}) \right). \nonumber
\end{align}
We also have $I(X; Y)-I(X; Y^{i-j})=$
\begin{align}
& p(y_{ij})
 H\left(\frac{p(y_i)P_{X|Y=y_j}+p(y_j)P_{X|Y=y_j}}{p(y_{ij})} \right) \nonumber\\
& - \left( p(y_i) H(P_{X|Y=y_i})+ p(y_j) H(P_{X|Y=y_j}) \right). \nonumber
\end{align}

\end{proposition}
\begin{proof}
 After merging $y_i$ and $y_j$, we obtain
\begin{align}
&p(s|y_{ij})= \frac{p(y_i)}{p(y_{ij})} p(s|y_i) + \frac{p(y_j)}{p(y_{ij})} p(s| y_j), \text{ for all } s\in \mathcal{S}, \nonumber \\
&p(x|y_{ij})= \frac{p(y_i)}{p(y_{ij})} p(x|y_i) + \frac{p(y_j)}{p(y_{ij})} p(x| y_j), \text{ for all } x\in \mathcal{X}. \nonumber
\end{align}
 The proof follows from writing $I(S; Y)-I(S; Y^{i-j})= H(S|Y^{i-j})- H(S|Y)$ and
$I(X; Y)-I(X; Y^{i-j})= H(X|Y^{i-j})- H(X|Y)$.
\end{proof}
Proposition \ref{Pro:JSD} shows that the difference in the mutual information after
merging changes only if the new variable, $y_{ij}$,  is involved.
 The greedy algorithm is locally optimal at every step since we minimize $I(S; Y)$. However,
there is no guarantee that such a greedy algorithm induces a global optimal privacy mapping.
\vspace{-.01in}
\begin{algorithm}[t] \label{alg:greedymax}
\caption{Greedy algorithm-information bottleneck}
\begin{algorithmic}
\STATE \textbf{Input:} $\Delta$, $P_{S, X}$
\STATE \textbf{Initialization: } $\mathcal{Y} = \mathcal{X}$, $P_{Y|X}(y|x)=\mathbf{1}\{y=x\}$
\WHILE{ there exists $i', j'$ such that $I(S; Y^{i'-j'}) \ge \Delta$}
 \STATE   among those $i', j'$, let \\ $\{y_i, y_j\}= \argmax_{y_{i'}, y_{j'} \in \mathcal{Y}}  I(X; Y) - I(X; Y^{i'-j'}) $
    \STATE    \textbf{merge}: $\{y_i, y_j\} \rightarrow y_{ij}$
    \STATE    \textbf{update}: $\mathcal{Y}= \{\mathcal{Y}\setminus \{y_i, y_j\} \}\cup \{y_{ij}\}$  and $P_{Y|X}$
\ENDWHILE
\STATE \textbf{Output:} $P_{Y|X}$
\end{algorithmic}
\vspace{-.06in}
\end{algorithm}
\vspace{-.05in}
\begin{note}
\textup{
The minimum of $I(S; Y)$ in \eqref{eq:optdkl} is a decreasing function of $I(X; Y)$ and is achieved for a mapping  $P_{Y|X}$ that satisfies $I(X; Y)=R$ (if possible due to discrete alphabets). For a given mutual  information, $R$, there are many conditional
probability distributions, $P_{Y|X}$, achieving  $I(X; Y)=R$.
Among which there is one that gives the minimum
$I(S; Y)$ and one that gives the maximum $I(S; Y)$.
We can modify the greedy algorithm so that it converges to a local maximum of $I(S; Y)$ for a given $I(X; Y)=R$.  The algorithm which we call \emph{greedy algorithm-information bottleneck} is given in Algorithm (2). Algorithm (1) and Algorithm (2) allow us to approximately characterize the range of values $I(S; Y)$ can take for a given value of $I(X; Y)$ as being those between the local minimum and the local maximum. Interestingly, by observing the gap between the local maximum and the local minimum, we have a relative idea on the effectiveness of the Greedy algorithm, i.e., if the difference is significant it means a negligent mapping may lie anywhere between those values, possibly leading to a much higher privacy threat.
}
\end{note}
\vspace{-.02in}
\subsection{Data Set}
 The US 1994 Census dataset \cite{Asuncion+Newman:2007} is a well-known dataset in the machine learning community, which is a sample of the US population from 1994. For each of the entries, it contains features such age, work-class, education, gender, and native country, as well as an income category. The income level is a binary variable which determines whether  the income is above or below USD 50000, gender is a binary variable, education level is a variable with four categories, age is a variable divided into seven categories. For our purposes, we consider the private attributes $S=(\text{age, income level})$ and the attributes to be released as $X=(\text{age, gender, education level})$.  The goal of the privacy mapping is to release a modified version of attributes $Y$ which is informative about $X$ but that renders the inference of $S$ based on $Y$ hard. 
%
\subsection{Numerical Results}
 In Fig.~\ref{ine3fig}, we plot the minimum and maximum of $I(S; Y)$ for a given $I(X; Y)$. This figure is based on US 1994 census data set described before. The top curve shows the maximum of $I(S; Y)$ versus $I(X; Y)$, using Algorithm (2). The bottom  curve shows the minimum of $I(S; Y)$ versus $I(X; Y)$, using Algorithm (1). The area  between the two curves shows the 
possible pairs of $(I(X; Y), I(S; Y))$ as $P_{Y|X}$  varies (a subset of possible pairs, since the algorithms are sub-optimal). Indeed, we will design the mapping to lie on the bottom curve. For  a given $R$, if we design the mapping negligently, we may have $I(S; Y)$ on the top curve  instead of the bottom curve.



\begin{figure}[t]
\centering
\includegraphics[width=0.45 \textwidth]{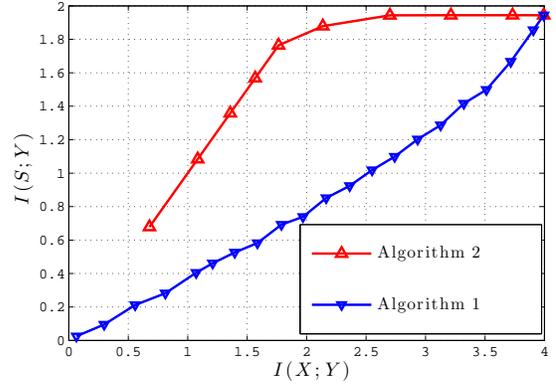}
\caption{Maximum and minimum of $I(S; Y)$ for a given $I(X; Y)$.}
\label{ine3fig}
\end{figure}
\vspace{-.03in}

\section{Conclusions}\label{sec:conclusion}
We study the privacy-utility trade-off against inference attacks when the log-loss is used both in the privacy and utility metrics. We justify the generality of the privacy threat under the log-loss by proving that the threat under any bounded cost inference function can be upperbounded by an explicit function of the mutual information between private and disclosed data. We cast the tradeoff under the log-loss as the Privacy Funnel optimization, which is non-convex. We leverage its connection to the Information Bottleneck to design a locally-optimal greedy algorithm, that we evaluate on the US census dataset.
\section*{Acknowledgement}
The authors are grateful to Prof. Thomas Courtade, Prof. Kave Salamatian, and Prof. Tsachy Weissman for encouraging them to study the connections between privacy and the information bottleneck.

\bibliographystyle{./biblio/IEEEtran}
\bibliography{./biblio/IEEEabrv,references}

\end{document}